\documentclass[letterpaper,11pt]{article}

\usepackage{times} 
\usepackage{fullpage} 
\usepackage{amsmath} 
\usepackage{amssymb} 
\usepackage{amsthm} 
\usepackage{bbm} 
\usepackage{tabularx} 
\usepackage[multiple]{footmisc} 

\usepackage{graphicx} 

\usepackage{ifpdf}
\ifpdf
\usepackage[
    pdftex,bookmarks,pagebackref,
    plainpages=false, 
        pdfpagelabels=true 
        ]{hyperref}
\else \fi

\let\oldFootnote\footnote
\newcommand\nextToken\relax

\renewcommand\footnote[1]{%
    \oldFootnote{#1}\futurelet\nextToken\isFootnote}

\newcommand\isFootnote{%
    \ifx\footnote\nextToken\textsuperscript{,}\fi}

\newcommand{\inner}[2]{\langle #1 , #2\rangle}
\newcommand{\Inner}[2]{\left\langle #1 , #2\right\rangle}

\newcommand{\cls}[1]{\mathrm{#1}}

\newtheorem{theorem}{Theorem}

\newenvironment{numberedtheorem}[1]
{
  \renewcommand{\thetheorem}{#1}
  \begin{theorem}
}{
  \end{theorem}
  \addtocounter{theorem}{-1}
}

\newtheorem{lemma}[theorem]{Lemma}
\newtheorem{corollary}{Corollary}[theorem]
\newtheorem{proposition}[theorem]{Proposition}

\theoremstyle{definition}
\newtheorem{defn}[theorem]{Definition}

\newtheorem{problem}{Problem}

\newcommand{\pa}[1]{(#1)}
\newcommand{\Pa}[1]{\left(#1\right)}

\newcommand{\set}[1]{\{#1\}}

\DeclareMathOperator{\trace}{Tr}
\newcommand{\ptr}[2]{\trace_{#1}\pa{#2}}
\newcommand{\Ptr}[2]{\trace_{#1}\Pa{#2}}

\DeclareMathOperator{\marginal}{mar}
\newcommand{\mar}[2]{\marginal_{#1}\pa{#2}}
\newcommand{\Mar}[2]{\marginal_{#1}\Pa{#2}}

\newcommand{\ol}[1]{{\overline{#1}}}
\newcommand{\ns}{\mathrm{ns}}
\def\ot{\otimes}

\def\cA{\mathcal{A}}
\def\cB{\mathcal{B}}

\def\cS{\mathcal{S}}
\def\cT{\mathcal{T}}

\def\cX{\mathcal{X}}
\def\cY{\mathcal{Y}}

\def\bA{\mathbf{A}}
\def\bB{\mathbf{B}}

\begin{document}

\title{Interactive proofs with competing teams of no-signaling provers}

\author{
  Gus Gutoski \\[3mm]
  {\small\it
  \begin{tabular}{c}
    Perimeter Institute for Theoretical Physics\\
    Waterloo, Ontario, Canada\thanks{This research was conducted while the author was a postdoc at the Institute for Quantum Computing and School of Computer Science at the University of Waterloo in Waterloo, Ontario, Canada.}
  \end{tabular}
  }
}

\date{December 3, 2010 \\[1mm] {\small (Minor revisions: October 14, 2011 and July 3, 2013)} }

\maketitle

\begin{abstract}

This paper studies a generalization of multi-prover interactive proofs in which a verifier interacts with two competing teams of provers: one team attempts to convince the verifier to accept while the other attempts to convince the verifier to reject.
Each team consists of two provers who jointly implement a \emph{no-signaling} strategy.
No-signaling strategies are a curious class of joint strategy that cannot in general be implemented without communication between the provers, yet cannot be used as a black box to establish communication between them.
Attention is restricted in this paper to \emph{two-turn} interactions in which the verifier asks questions of each of the four provers and decides whether to accept or reject based on their responses.

We prove that the complexity class of decision problems that admit two-turn interactive proofs with competing teams of no-signaling provers is a subset of $\cls{PSPACE}$.
This upper bound matches existing $\cls{PSPACE}$ lower bounds on the following two disparate and weaker classes of interactive proof:
\begin{enumerate}

\item
Two-turn multi-prover interactive proofs with only one team of no-signaling provers.

\item
Two-turn competing-prover interactive proofs with only one prover per team.

\end{enumerate}
Our result implies that the complexity of these two models is unchanged by the addition of a second competing team of no-signaling provers in the first case and by the addition of a second no-signaling prover to each team in the second case.
Moreover, our result unifies and subsumes prior $\cls{PSPACE}$ upper bounds on these classes.

\end{abstract}

\section{Introduction}

Interactive proofs were introduced in the mid-1980's as a generalization of the concept of efficient proof verification and the complexity class $\cls{NP}$ \cite{Babai85,BabaiM88,GoldwasserM+89}.
Informally speaking, an \emph{interactive proof} is a conversation between a randomized polynomial-time \emph{verifier} and a computationally unbounded \emph{prover} regarding some common input string $x$.
A decision problem $L$ is said to admit an interactive proof if there exists a verifier such that (i) if $x$ is a yes-instance of $L$ then there is a prover who can convince the verifier to accept $x$ with high probability, and (ii) if $x$ is a no-instance of $L$ then no prover can convince the verifier to accept $x$ except with small probability.
In a dramatic testament to the surprising power of randomization and interaction, it was soon discovered that every problem in $\cls{PSPACE}$ admits an interactive proof, yielding the well-known identity $\cls{IP}=\cls{PSPACE}$ \cite{LundF+92,Shamir92}.

\subsubsection*{Multi-prover interactive proofs, no-signaling provers}

The fruitful study of interactive proofs has prompted further generalization of the model.
One such generalization is the \emph{multi-prover} interactive proof model of Ben-Or \emph{et al.}~\cite{Ben-OrG+88} wherein several provers cooperate in their attempt to convince the verifier to accept the input string $x$.
The key aspect that sets this model apart from single-prover interactive proofs is the fact that the provers cannot communicate with one another during the protocol.
Amazingly, this small distinction is enough to increase the power of the model from $\cls{PSPACE}$ all the way up to $\cls{NEXP}$ \cite{BabaiF+91,FortnowR+94}, even when the interaction is restricted to only two turns with only two provers \cite{FeigeL92}.
In terms of complexity classes, the corresponding identity is $\cls{MIP}=\cls{NEXP}$.

Intermediate classes of multi-prover interactive proofs are obtained by tinkering with the set of strategies available to the provers.
Consider, for example, a joint strategy where the distribution of answers from one prover is independent of the question asked of the other prover---these are the \emph{no-signaling} strategies.
Clearly, such a strategy cannot be used in a black-box fashion by the provers to establish communication.
At first glance it may seem that the no-signaling condition is equivalent to the standard definition of a multi-prover interactive proof.
However, there exist no-signaling strategies that cannot be implemented without communication between the provers, suggesting that this model might be a nontrivial intermediary between single- and multi-prover interactive proofs.

Indeed, it was established by Ito, Kobayashi, and Matsumoto \cite{ItoK+09} that the two-turn, two-prover protocol for $\cls{PSPACE}$ of Cai, Condon, and Lipton \cite{CaiC+94} is sound even against no-signaling provers.
By contrast, $\cls{PSPACE}$ is known not to admit two-turn \emph{single}-prover interactive proofs unless the polynomial hierarchy collapses and $\cls{PSPACE}=\cls{AM}$ \cite{Babai85,GoldwasserS89}.
A converse result was proven by Ito, who showed that every problem that admits a two-turn interactive proof with two no-signaling provers is also in $\cls{PSPACE}$ \cite{Ito09}.
Thus, the interactive proof model is even \emph{more} sensitive to change than suggested by the difference between single- and multi-prover interactive proofs, as even the smaller difference between no-signaling and standard multi-prover interactive proofs is sufficient to make the jump from $\cls{PSPACE}$ up to $\cls{NEXP}$ (at least in the case of two turns and two provers).

In addition to this prior work, parallel repetition results for multi-prover interactive proofs with no-signaling provers were established in Refs.~\cite{Holenstein09,KempeR10}.
The reader is referred to Ito \cite{Ito09} for more detailed history and references.

\subsubsection*{Inspiration from quantum information}

Though the present paper contains no formal discussion of quantum information, it is proper to acknowledge its role in motivating the study of no-signaling provers.
Interest in this model was originally drawn from the study of multi-prover \emph{quantum} interactive proofs, in which the provers (and possibly the verifier) are permitted to exchange and manipulate quantum information.

It is easy to see that interactive proofs with ordinary, ``classical'' provers are not affected by the ability of the provers to sample from a common source of randomness.
Quantum provers, on the other hand, might use shared pieces of some entangled quantum state to implement a \emph{nonlocal} strategy that correlates their messages in ways that cannot otherwise be achieved.
(The phenomenon of nonlocality was famously branded by Einstein as ``spooky action at a distance.'')
Indeed, some classical protocols which are sound against classical provers are known to become unsound when the provers share entanglement \cite{CleveH+04,CleveGJ07}.
Incredibly, the class $\cls{MIP}^*$ of decision problems that admit multi-prover interactive proofs with entanglement-sharing provers is not even known to be computable, a consequence of the fact that no bound is known on the amount of entanglement needed to approximate an optimal strategy.

Whereas the set of entanglement-sharing strategies is highly complex, the set of no-signaling strategies is relatively simple and it includes entanglement-sharing strategies as a proper subset.
So, for example, any protocol that is sound against no-signaling provers is also sound against quantum provers who share entanglement.
It is also interesting to find differences between no-signaling strategies and entanglement-sharing strategies, as this difference sheds light on the extent to which no-signaling can be used as a proxy for shared entanglement.
In some protocols the allowance of arbitrary no-signaling strategies leads to implausible consequences \cite{vanDam05,BrassardBL+06}.
Such protocols can be viewed as mathematical evidence against physical theories that admit so-called ``super-strong'' nonlocality such as that found in no-signaling strategies but not entanglement-sharing strategies.
The present paper establishes a scenario in which two no-signalling provers are equivalent to two signaling provers.

\subsubsection*{Interactive proofs with competing provers}

Another generalization of the single-prover model is an interactive proof with \emph{competing provers}, in which one prover tries to convince the verifier to accept the input string $x$ while the other prover tries to convince the verifier to reject $x$.
One may consider proofs in which all messages are known to all provers (\emph{complete} information) or in which each prover sees only the messages he exchanges with the verifier (\emph{incomplete} information).
These two forms of competing-prover interactive proofs were studied by several authors in the 1990's \cite{FeigeS+90,FeigeS92,FeigenbaumK+95,FeigeK97}.
But for our purpose in this paper it only makes sense to consider protocols with incomplete information.

In the jargon of game theory, interactive proofs with competing provers are \emph{zero-sum games}, about which there exists a vast body of literature in computer science, economics, and other disciplines.
For instance, fast algorithms for zero-sum games of incomplete information in \emph{extensive form} imply that the complexity class $\cls{RG}$ of problems that admit interactive proofs with competing provers is a subset of $\cls{EXP}$ \cite{KollerM92,KollerMvS94}.
Feige and Kilian proved the reverse containment \cite{FeigeK97}, yielding the competing-prover analogy $\cls{RG}=\cls{EXP}$ of the aforementioned identity $\cls{IP}=\cls{PSPACE}$ for single-prover interactive proofs.

Feige and Kilian also studied \emph{two-turn} interactive proofs with competing provers, providing a matching upper and lower bound of $\cls{PSPACE}$ on the complexity of this model \cite{FeigeK97}.
The complexity of $k$-turn interactive proofs with competing provers
for constants $k\geq 3$ is an open question of interest to both complexity theorists and game theorists alike.

\subsubsection*{Interactive proofs with competing teams of provers, our result}

Multi-prover interactive proofs and interactive proofs with competing provers are two distinct generalizations of the single-prover model.
The next logical step is to unify these two generalizations in the obvious way via interactive proofs with competing \emph{teams} of provers.
Combining established naming conventions for complexity classes based on interactive proofs, we let $\cls{MRG}$ denote the class of decision problems that admit interactive proofs with competing teams of provers.

To the author's knowledge, this model was considered prior to the present work only by Feigenbaum, Koller, and Shor \cite{FeigenbaumK+95}.
Those authors studied this class under the game-theoretic guise of zero-sum games of \emph{imperfect recall} and proved the containments
\[ \cls{EXP^{NP}} \subseteq \cls{MRG} \subseteq \cls{\Sigma^{EXP}_2} \cap \cls{\Pi^{EXP}_2} \]
where $\cls{\Sigma^{EXP}_2}$ and $\cls{\Pi^{EXP}_2}$ are classes in the second level of the exponential hierarchy, which is the exponential-time version of the familiar polynomial hierarchy.

In this paper we consider interactive proofs with competing teams of \emph{no-signaling} provers.
Our main result is as follows.

\begin{theorem}[$\cls{MRG}_\ns(2,2)=\cls{PSPACE}$]
\label{thm:main-result}

Every decision problem that admits a two-turn interactive proof with competing teams of two no-signaling provers per team is also in $\cls{PSPACE}$.
Letting $\cls{MRG_{ns}}(2,2)$ denote the complexity class of such problems, it follows that $\cls{MRG}_\ns(2,2)=\cls{PSPACE}$.

\end{theorem}

This upper bound matches the aforementioned $\cls{PSPACE}$ lower bounds on the following two disparate and weaker classes of interactive proof:
\begin{enumerate}

\item
Two-turn multi-prover interactive proofs with only one team of no-signaling provers \cite{CaiC+94,ItoK+09}.

\item
Two-turn competing-prover interactive proofs with only one prover per team \cite{FeigeK97}.

\end{enumerate}
Our result implies that the complexity of these two models is unchanged by the addition of a second competing team of no-signaling provers in the first case and by the addition of a second no-signaling prover to each team in the second case.
Moreover, our result unifies and subsumes prior $\cls{PSPACE}$ upper bounds on these classes \cite{Ito09,FeigeK97}.

Interactive proofs with competing teams of no-signaling provers were not considered prior to the present work.
As such, no explicit upper bound on $\cls{MRG}_\ns(2,2)$ has been observed until now.
A trivial upper bound for this class, as for $\cls{MRG}$, is the second level of the exponential hierarchy.

\subsubsection*{Techniques}

Theorem \ref{thm:main-result} is proven by means of an efficient parallel algorithm that, given an explicit description of a verifier and an accuracy parameter $\delta$, finds no-signaling strategies for the teams that are within $\delta$ of optimal.
Containment in $\cls{PSPACE}$ then follows in the usual way by observing that the description of the verifier has size exponential in the length of the input string $x$ and then employing the fact that a parallel algorithm with succinct input can be simulated in polynomial space \cite{Borodin77}.

Our algorithm is an example of the \emph{multiplicative weights update method (MWUM)} as discussed in the survey paper \cite{v008a006} and in the PhD thesis of Kale \cite{Kale07}.
(See also Ref.\ \cite{WarmuthK06}.)
In its simplest form, the MWUM solves a min-max optimization problem on probability distributions.
In the present paper we use the MWUM to optimize not just a \emph{single} distribution, but many distributions \emph{simultaneously} in the form of a stochastic matrix that represents a strategy for one of the teams.
This trick
seems
to work only for two-turn protocols, as otherwise it is not clear how to ensure sufficient accuracy.

Let us compare our algorithm to the two previous algorithms it subsumes:
\begin{itemize}

\item
The polynomial-space algorithm of Feige and Kilian for two-turn interactive proofs with competing provers \cite{FeigeK97} is a complicated and highly specialized precursor to the MWUM that, like our algorithm, optimizes over stochastic matrices that represent strategies for the provers.

Their algorithm works by nondeterministically guessing the entries of the matrix and scanning them in a read-once fashion.
This approach cannot be extended to optimize over no-signaling strategies, as the read-once model does not
allow verification of the no-signaling condition.

\item
The parallel algorithm of Ito for two-turn, two-prover interactive proofs with no-signaling provers \cite{Ito09} is essentially a reduction to the \emph{mixed packing and covering problem}, which is a special type of linear program that
is known to admit an efficient parallel algorithm \cite{Young01}.

This approach, too, cannot be extended to competing teams of no-signaling provers, as any linear programming formulation of the
protocol is unlikely to be a mixed packing and covering problem.

\end{itemize}

Our study has benefitted
from the valuable experience of recent applications of the MWUM to parallel algorithms for quantum complexity classes \cite{JainW08,JainU+09,JainJ+11,Wu10,GutoskiW13-invited}.
Indeed, we follow the same high-level approach as the recent proof of $\cls{DQIP}=\cls{DIP}=\cls{PSPACE}$ \cite{GutoskiW13-invited}.
Namely,
\begin{itemize}

\item
The domain of admissible (no-signaling) strategies is a strict subset of the ``natural'' domain (stochastic matrices) for the MWUM.

\item
To get around this problem, the strategy domain is extended to \emph{all} the stochastic matrices and a \emph{penalty term} is introduced so as to remove any incentive for a team to use an inadmissible strategy.
(See Section \ref{sec:relaxation}).

\item
Finally, one must prove a ``rounding'' theorem (Corollary \ref{cor:rounding}), which establishes that near-optimal, fully-admissible strategies can be obtained from near-optimal strategies in the extended domain with penalty term.

\end{itemize}
Section \ref{sec:conclusion} describes some difficulties that arise when this approach is applied to complexity classes beyond $\cls{MRG}_\ns(2,2)$ and lists some open problems.

\section{Preliminaries}
\label{sec:defs}

\subsection{Definition of two-turn interactive proofs with competing teams of provers}
\label{sec:defs:mrg}

In this paper we are concerned with decision problems that admit two-turn interactive proofs with competing teams of no-signaling provers.
Let us clarify this concept.
A \emph{two-turn verifier} is a randomized polynomial-time algorithm that, given an input string $x$, produces questions $i,j$ for the two teams of provers.
The teams select their answers $k,l$ (possibly using randomness to do so) and then the verifier accepts or rejects the input $x$ according to some boolean function of $i,j,k,l$.
For convenience, the teams shall be called \emph{Team Alice} and \emph{Team Bob}.
It is the goal of Team Alice to convince the verifier to accept the input string $x$, while Team Bob's goal is to convince the verifier to reject $x$.

In the protocols we consider each team consists of two provers.
The provers of Team Alice shall be called \emph{Alice$_0$} and \emph{Alice$_1$}, while the provers of Team Bob shall be called \emph{Bob$_0$} and \emph{Bob$_1$}.
Each individual prover on each team receives his or her own private question and supplies his or her own separate answer to the verifier.
In particular, the question $i$ asked of Team Alice is actually a pair $i=(i_0,i_1)$ with question $i_c$ going to prover Alice$_c$ for both values of the bit $c\in\set{0,1}$.
Similarly, the question $j$ asked of Team Bob is also a pair $j=(j_0,j_1)$ with question $j_c$ going to prover Bob$_c$.
The answers $k,l$ received from the two teams are also pairs $k=(k_0,k_1)$ and $l=(l_0,l_1)$ with answers $k_c$ and $l_c$ coming from Alice$_c$ and Bob$_c$, respectively.
The entire interaction is illustrated in Figure \ref{fig:mrg22}.

\begin{figure}
  \begin{center}
\includegraphics{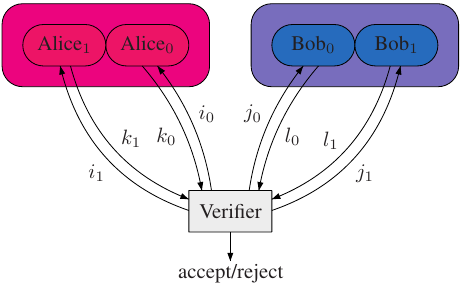}
  \end{center}
  \caption{A two-turn interactive proof with competing teams of two no-signaling provers per team.}
  \label{fig:mrg22}
\end{figure}

Each team may jointly implement any no-signaling strategy in order to produce its answers.
Briefly, a strategy for, say, Team Alice is \emph{no-signaling} if the marginal distribution on answers $k_0$ from Alice$_0$ does not depend upon the question $i_1$ asked of Alice$_1$ and \emph{vice versa}.
No-signaling strategies are discussed in greater detail in Section \ref{sec:defs:no-sig}.

A decision problem $L$ is said to admit a two-turn interactive proof with competing teams of no-signaling provers with \emph{completeness} $c$ and \emph{soundness} $s$ if there exists a fixed two-turn verifier with the following properties:
\begin{description}

\item[Completeness.]
If the input string $x$ is a yes-instance of $L$ then there exists a no-signaling strategy for Team Alice that convinces the verifier to accept $x$ with probability at least $c$, regardless of the no-signaling strategy employed by Team Bob.

\item[Soundness.]
If the input string $x$ is a no-instance of $L$ then there exists a no-signaling strategy for Team Bob that convinces the verifier to reject $x$ with probability at least $1-s$, regardless of the no-signaling strategy employed by Team Alice.

\end{description}
The completeness and soundness parameters need not be fixed constants.
Rather, they may vary as a function of the input string $x$.
The complexity class $\cls{MRG}_\ns(2,2)$ consists of all decision problems that admit two-turn interactive proofs with competing teams of two no-signaling provers per team with completeness $c$ and soundness $s$ such that there exists a fixed polynomial-bounded function $p$ on strings with $c-s\geq 1/p$.
(The first parameter of the class $\cls{MRG}_\ns(2,2)$ denotes the number of provers per team, the second denotes the number of turns in the protocol.
It is also common to parameterize interactive proof classes according to the number of \emph{rounds} of communication, rather than the number of \emph{turns}.
Under this scheme, the class $\cls{MRG}_\ns(2,2)$ might be called $\cls{MRG}_\ns(2,1)$ by some authors.)

In this paper we prove $\cls{MRG}_\ns(2,2)\subseteq\cls{PSPACE}$ (Theorem \ref{thm:main-result}).
It then follows from existing lower bounds on weaker classes \cite{ItoK+09,FeigeK97} that \( \cls{MRG}_\ns(2,2)=\cls{PSPACE}. \)

\subsection{Notation, the Kronecker product}
\label{sec:defs:notation}

To each interactive proof with input $x$ we associate eight distinct finite-dimensional real Euclidean spaces---four \emph{question} spaces and four \emph{answer} spaces.
These spaces are denoted as follows for both $c\in\set{0,1}$:
\begin{align*}
  \cS_c \quad & \textrm{The question space for prover Alice$_c$}
& \cA_c \quad & \textrm{The answer space for prover Alice$_c$}\\
  \cT_c \quad& \textrm{The question space for prover Bob$_c$}
& \cB_c \quad& \textrm{The answer space for prover Bob$_c$}
\end{align*}
The dimension of each space is the number of distinct questions or answers available to that prover.
(For example, prover Alice$_0$ can be asked any of $\dim(\cS_0)$ distinct questions and may respond with any of $\dim(\cA_0)$ distinct answers.)
Individual questions or answers are indexed by positive integers denoted for both $c\in\set{0,1}$ as follows:
\begin{align*}
  \textrm{Questions for Alice$_c$}: \quad & i_c=1,\dots,\dim(\cS_c)\\
  \textrm{Questions for Bob$_c$}  : \quad & j_c=1,\dots,\dim(\cT_c)\\
  \textrm{Answers from Alice$_c$} : \quad & k_c=1,\dots,\dim(\cA_c)\\
  \textrm{Answers from Bob$_c$}   : \quad & l_c=1,\dots,\dim(\cB_c)
\end{align*}
Since the verifier acts in polynomial time, the bit length of the questions and answers is at most a polynomial in the bit length $|x|$ of the input string $x$.
Since $n$ bits suffice to encode $2^n$ distinct questions or answers, the dimension of the spaces $\cS_c,\cT_c,\cA_c,\cB_c$ can be exponential in $|x|$.

The \emph{Kronecker product} (or \emph{tensor product}) of two spaces $\cX,\cY$ is another space with dimension $\dim(\cX)\dim(\cY)$.
This product space is typically denoted by $\cX\ot\cY$, which we abbreviate to $\cX\cY$.
Kronecker products involving the eight spaces $\cS_c,\cT_c,\cA_c,\cB_c$ are further abbreviated so that \[\cS_{01}=\cS_0\cS_1=\cS_0\ot\cS_1\] and so on.
The Kronecker product extends in a natural way to vectors and linear operators.
In this paper each vector or linear operator is implicitly associated with its representation as a column or a matrix, for which the Kronecker product is given by a straightforward formula.
For example, if $A,B$ are $2\times 2$ matrices given by
\[
  A= \left[ \begin{array}{cc} a & b \\ c & d \end{array} \right]
  ,\qquad
  B= \left[ \begin{array}{cc} p & q \\ r & s \end{array} \right]
\]
then the Kronecker product $A\ot B$ is given by
\[
  A\ot B =
  \left[ \begin{array}{cc} aB & bB \\ cB & dB \end{array} \right]
  =
  \left[ \begin{array}{cc}
      a \left[ \begin{array}{cc} p & q \\ r & s \end{array} \right]
    & b \left[ \begin{array}{cc} p & q \\ r & s \end{array} \right] \\
      c \left[ \begin{array}{cc} p & q \\ r & s \end{array} \right]
    & d \left[ \begin{array}{cc} p & q \\ r & s \end{array} \right]
  \end{array} \right]
  =
  \left[
    \begin{array}{cccc}
    ap & aq & bp & bq \\
    ar & as & br & bs \\
    cp & cq & dp & dq \\
    cr & cs & dr & ds
    \end{array}
  \right].
\]
This definition extends in the obvious way to arbitrary matrices of any dimension, including column vectors and other non-square matrices.

We also make use of the following symbols:
\begin{center}
\begin{tabularx}{\textwidth}{lX}
  $e_\cX$ & The all-ones column vector of dimension $\dim(\cX)$. \\
  $I_\cX$ & The identity matrix acting on $\cX$. \\
  $M^*$ & The \emph{adjoint} of a linear mapping $M$.
    If $M$ is a matrix or column vector then $M^*$ is simply the transpose of $M$.\\
  $\inner{A}{B}$ & The \emph{matrix inner product}, defined as $\ptr{}{A^*B}$.
    This inner product is defined only when the dimensions of $A,B$ are equal.
    If $A,B$ are vectors then $\inner{A}{B}$ is called the \emph{vector inner product}.\\
  $\leq,\geq$ & Matrix inequalities are entrywise.\\
  $\ol{c}$ & Given a bit $c\in\set{0,1}$, the compliment $\ol{c}$ is given by $\ol{c}=1$ if $c=0$, otherwise $\ol{c}=0$.
\end{tabularx}
\end{center}

\subsection{Min-max formalism for interactive proofs with competing provers}
\label{sec:defs:min-max}

Given a fixed two-turn verifier and a fixed input string $x$, let $\pi_{i,j}$ denote the probability with which the verifier asks questions $i=(i_0,i_1)$ to Team Alice and $j=(j_0,j_1)$ to Team Bob.
For each 4-tuple $(i,j)$ of questions to the provers let $v_{i,j} \in\cA_{01}\cB_{01}$ denote the 0-1 vector of \emph{payouts} to Team Bob.
That is, for each $k=(k_0,k_1)$ and each $l=(l_0,l_1)$ the $(k,l)$th entry of $v_{i,j}$ is either zero or one according to whether the verifier accepts or rejects $x$ in the event that the verifier asks questions $(i,j)$ to the teams and they respond with answers $(k,l)$.\footnote{
  One could consider a more general referee in which the payouts are awarded probabilistically so that each entry of $v_{i,j}$ lies in the interval $[0,1]$.
  But it is easily seen that this model is equivalent to the one we have just described.}\footnote{
  The payout vector $v_{i,j}$ is defined so that 0 indicates acceptance of $x$ while 1 indicates rejection.
  This arbitrary choice is opposite of convention, but it better facilitates the forthcoming presentation of our multiplicative weights update algorithm.}
Consider the entrywise nonnegative matrix \[ V:\cS_{01}\cT_{01}\to\cA_{01}\cB_{01} \]
whose $(i,j)$th column is $\pi_{i,j} v_{i,j}$.
This matrix uniquely specifies the actions of the verifier.

Strategies for the teams are specified as follows.
For each pair $i$ of questions let $a_i\in\cA_{01}$ denote the probability vector of Team Alice's responses to $i$.
That is, for each pair $k$ of answers the $k$th entry of $a_i$ denotes the probability with which Team Alice replies with answers $k$ given that questions $i$ were asked.
Thus, the actions of Team Alice are uniquely specified by the stochastic matrix
\[ A:\cS_{01}\to\cA_{01} \]
whose $i$th column is $a_i$.
Similarly, for each pair $j$ of questions let $b_j\in\cB_{01}$ denote the probability vector of Team Bob's responses to $j$.
The actions of Team Bob are uniquely specified by the stochastic matrix
\[ B:\cT_{01}\to\cB_{01} \]
whose $j$th column is $b_j$.
Not every stochastic matrix denotes a valid no-signaling strategy for the teams.
Criteria for no-signaling strategies are discussed in Section \ref{sec:defs:no-sig}.
For now, it suffices to note that the set of all strategies available to each team is a compact convex subset of stochastic matrices.

Conditioned on the verifier asking questions $(i,j)$, it is clear that the probability of rejection is given by the vector inner product
\[ \Inner{v_{i,j}}{a_i\ot b_j}. \]
It follows that the probability of rejection---taken over all questions $(i,j)$---given strategies $A$ for Team Alice and $B$ for Team Bob is given by the matrix inner product
\[
  \Pr[\textrm{$V$ rejects $x$} \mid A,B] = \inner{V}{A\ot B}
  = \sum_{i,j}
  \pi_{i,j} \Inner{v_{i,j}}{a_i\ot b_j}.
\]

Of course, Team Bob wishes to maximize this quantity while Team Alice wishes to minimize this quantity.
Given that the above inner product is bilinear in $(A,B)$ and that the sets of admissible strategies for the two teams are compact and convex, it follows from standard min-max theorems \cite{Ville38,Fan53} that every interactive proof with verifier $V$ has an \emph{equilibrium value}, which we denote by $\lambda(V)$, given by
\[
  \lambda(V)
  = \min_A \max_B \inner{V}{A\ot B}
  = \max_B \min_A \inner{V}{A\ot B}
\]
where the minimum is over all no-signaling matrices $A:\cS_{01}\to\cA_{01}$ and the maximum is over all no-signaling matrices $B:\cT_{01}\to\cB_{01}$.
In particular, for every protocol there exists at least one \emph{equilibrium point} $(A^\star,B^\star)$ with the property that
\begin{align*}
  \inner{V}{A^\star\ot B} &\leq \lambda(V) \quad \textrm{for all $B$}, \\
  \inner{V}{A\ot B^\star} &\geq \lambda(V) \quad \textrm{for all $A$}.
\end{align*}
Thus, the strategy $B^\star$ always ensures maximum likelihood of rejection, while $A^\star$ always ensures minimum likelihood of rejection.

This min-max theorem applies to every min-max expression considered throughout this paper.
Henceforth we do not bother to explicitly remark upon this fact.
Here and throughout the paper we adopt the convention that for any min-max problem of the form
\[ \nu(g) = \min_{a\in\bA}\max_{b\in\bB} g(a,b) \]
elements $\tilde a\in\bA$ and $\tilde b\in\bB$ are \emph{$\delta$-optimal} if
\begin{align*}
  g(\tilde a,b) &\leq \nu(g) + \delta \quad \textrm{for all $b\in\bB$}, \\
  g(a,\tilde b) &\geq \nu(g) - \delta \quad \textrm{for all $a\in\bA$}.
\end{align*}
Elements that are $0$-optimal---such as $A^\star,B^\star$ above---are simply called \emph{optimal}.

\subsection{Notation for marginal distributions}
\label{sec:defs:mar}

Before we discuss no-signaling strategies in detail it is beneficial to introduce notation for marginal probability distributions
that will be used throughout the remainder of this paper.
Suppose, for instance, that $a\in\cA_{01}$ is a probability vector of answers from Team Alice to some question from the verifier.
We let $\mar{\cA_1}{a}\in\cA_0$ denote the probability vector for the marginal distribution on answers from the prover Alice$_0$.
Basic probability theory dictates that the mapping $\marginal_{\cA_1}$ satisfy
\[ \textrm{$k_0$th entry of $\mar{\cA_1}{a}$} \equiv \sum_{k_1=1}^{\dim(\cA_1)} \textrm{$(k_0,k_1)$th entry of $a$}. \]
Of course, this mapping may be extended to arbitrary real vectors.
For arbitrary spaces $\cX,\cY$ the linear mapping $\marginal_\cY$ is defined by
\[ \marginal_\cY : \cX\cY \to \cX : x\ot y\mapsto \inner{e_\cY}{y}x. \]
(The matrix representation of $\marginal_\cY$ is $e_\cY^*\ot I_\cX$.)
While this mapping is primarily intended to denote marginal probability distributions, we will have occasion to use it on non-probability vectors in this paper.

The mapping $\marginal_\cY$ is to vectors as the partial trace is to square matrices.
The \emph{partial trace} $\trace_\cY$ is a linear mapping defined by
\[ \trace_\cY:X\ot Y \mapsto \ptr{}{Y} X \]
for any square matrices $X:\cX\to\cX$ and $Y:\cY\to\cY$.
Readers familiar with quantum information know that the state of a quantum register can be computed from a joint state of several registers via the partial trace.
So too with probability distributions: the distribution on states of a classical register can be computed from a joint distribution on states of several registers via $\marginal_\cY$.

The mapping $\marginal_\cY$ extends naturally from vectors to matrices by applying $\marginal_\cY$ to each column:
\[ \textrm{$i$th column of $\mar{\cY}{A}$} \equiv \Mar{\cY}{\textrm{$i$th column of $A$}}. \]
So, for example, if Team Alice acts according to the stochastic matrix $A$ then the stochastic matrix \[ \mar{\cA_1}{A} : \cS_{01}\to\cA_0 \] describes the ``marginal'' strategy for prover Alice$_0$.
That is, the $(i_0,i_1)$th column of $\mar{\cA_1}{A}$ is the distribution on answers $k_0$ from Alice$_0$ given questions $(i_0,i_1)$ from the verifier.

\subsection{Characterization of no-signaling strategies}
\label{sec:defs:no-sig}

Recall that a strategy for Team Alice is \emph{no-signaling} if for both values of the bit $c\in\set{0,1}$ the marginal distribution on answers $k_c$ from Alice$_c$ does not depend on the question $i_\ol{c}$ asked of Alice$_\ol{c}$.

In terms of Team Alice's stochastic matrix $A$, this condition means that for each $i_c$ the $(i_0,i_1)$th column of $\mar{\cA_\ol{c}}{A}$ is identical for all subindices $i_\ol{c}$.
Letting $a_{i_c}$ denote this fixed probability vector and letting $A_c:\cS_c\to\cA_c$ denote the stochastic matrix whose columns are $a_{i_c}$, the above condition can be written as
\[ \mar{\cA_c}{A} = A_c \ot e_{\cS_\ol{c}}^*. \]
We have just proven the following simple proposition.

\begin{proposition}[Characterization of no-signaling strategies]
\label{prop:strategy}

A stochastic matrix $A:\cS_{01}\to\cA_{01}$ denotes a no-signaling strategy for Team Alice if and only if for both values of the bit $c\in\set{0,1}$ there exists a stochastic matrix $A_c:\cS_c\to\cA_c$ such that
\[  \mar{\cA_\ol{c}}{A} = A_c \ot e_{\cS_\ol{c}}^*.  \]
A similar characterization holds for Team Bob.

\end{proposition}

Stochastic matrices $A$ meeting this condition are called \emph{no-signaling matrices}.
The matrices $A_c$ are said to \emph{witness} the fact that $A$ is a no-signaling matrix.
It follows immediately from Proposition \ref{prop:strategy} that the set of all no-signaling strategies available to each team is compact and convex---a fact already used in Section \ref{sec:defs:min-max} to assert the existence of optimal strategies for the teams.

\section{A relaxed min-max problem with penalties}
\label{sec:relaxation}

As mentioned in the introduction, the MWUM in its simplest form solves min-max optimization problems over probability vectors.
We optimize over stochastic matrices for the teams by using the MWUM simultaneously on each column of these matrices---a trick that works only for two-turn protocols, as we shall soon see.

We noted in Section \ref{sec:defs:no-sig} that the no-signaling matrices available to the teams form a strict subset of the stochastic matrices.
In order to optimize only over no-signaling matrices, in this section we specify a new min-max optimization problem $\mu(V)$ in which the teams may use \emph{arbitrary} strategies but pay a \emph{penalty} for strategies that violate the no-signaling condition.
By a careful choice of penalty, we remove the incentive of the teams to select inadmissible strategies without ruining the precarious convergence properties of the MWUM.

Some preliminary observations are given in Section \ref{sec:relaxation:bounds} before the formal definition of the new min-max problem $\mu(V)$ in Section \ref{sec:relaxation:mu}.
Equivalence of $\mu(V)$ and $\lambda(V)$ is proven in Section \ref{sec:relaxation:equivalence} with proofs of some lemmas in Section \ref{sec:relaxation:lemmas}.

\subsection{Bounds on two-turn verifiers}
\label{sec:relaxation:bounds}

First, for ease of notation we let $\Phi_V$ denote the unique linear transformation satisfying
\[
  \Inner{V}{A\ot B} =
  \Inner{\Phi_V(A)}{B} =
  \Inner{A}{\Phi_V^*(B)}
\]
for all matrices $A,B$.
Though a precise formula for $\Phi_V$ is of little use in this paper, for completeness we note that
\begin{align*}
  \Phi_V(A) &= \Ptr{\cS_{01}}{\Pa{A^*\ot I_{\cB_{01}}}V} \\
  \Phi_V^*(B) &= \Ptr{\cT_{01}}{\Pa{I_{\cA_{01}}\ot B^*}V}
\end{align*}
where $\trace_{\cS_{01}}$ and $\trace_{\cT_{01}}$ are partial trace mappings mentioned in Section \ref{sec:defs:mar}.
At the risk of hijacking terminology from functional analysis, the matrix $\Phi_V(A)$ can be viewed as a \emph{partial inner product} between $V$ and $A$.
This matrix can also be viewed as a new two-turn verifier for Team Bob obtained by ``hard-wiring'' Team Alice's strategy $A$ into the original verifier $V$.

Next, let $\pi\in\cS_{01}\cT_{01}$ denote the probability vector for the distribution on questions asked by the verifier.
In the notation of Section \ref{sec:defs:min-max}, the $(i,j)$th entry of $\pi$ is $\pi_{i,j}$---the probability with which the verifier asks questions $i$ to Team Alice and $j$ to Team Bob.
Let $\pi_\textrm{Alice}\in\cS_{01}$ denote the marginal distribution
\[ \pi_\textrm{Alice} = \mar{\cT_{01}}{\pi} \]
on questions to Team Alice, so that the $i$th entry $\pi_i$ of $\pi_\textrm{Alice}$ is $\pi_i\equiv\sum_j \pi_{i,j}$.
It is not hard to see that \[ V \leq e_{\cA_{01}\cB_{01}} \pi^* \] with equality achieved in the extreme case that each of the verifier's payout vectors $v_{i,j}$ is equal to the all-ones vector $e_{\cA_{01}\cB_{01}}$.
(Recall that matrix inequalities are entrywise.)
Similarly, it is easy to prove analogous inequalities for $\Phi_V(A),\Phi_V^*(B)$.
For example:

\begin{proposition}
\label{prop:Phi-bound}

For any stochastic matrix $B:\cT_{01}\to\cB_{01}$ it holds that
\( \Phi_V^*(B) \leq e_{\cA_{01}}\pi_\textrm{Alice}^*. \)

\end{proposition}

\begin{proof}

Let $A:\cS_{01}\to\cA_{01}$ be any nonnegative matrix and let $a_i,b_j$ denote the columns of $A,B$, respectively.
Then
\[
  \Inner{A}{\Phi_V^*(B)} =
  \Inner{V}{A\ot B} \leq
  \Inner{e_{\cA_{01}\cB_{01}}\pi^*}{A\ot B} =
  \sum_{i,j} \pi_{i,j} \Inner{e_{\cA_{01}}}{a_i} \Inner{e_{\cB_{01}}}{b_j}
\]
As $B$ is stochastic it must be that $\Inner{e_{\cB_{01}}}{b_j}=1$ for each $j$.
The above expression then simplifies to
\[ \sum_i \Pa{\sum_j \pi_{i,j}} \Inner{e_{\cA_{01}}}{a_i} = \Inner{e_{\cA_{01}}\pi_\textrm{Alice}^*}{A}. \]
As this inequality holds for all nonnegative matrices $A$ it must be that
\( \Phi_V^*(B) \leq e_{\cA_{01}}\pi_\textrm{Alice}^* \)
as claimed.
\end{proof}

\subsection{Definition of the relaxed min-max problem}
\label{sec:relaxation:mu}

 The relaxation $\mu(V)$ of $\lambda(V)$ is defined by
\[ \mu(V) = \min_{(A,A_0,A_1)} \max_{(B,\Pi_0,\Pi_1)} \Inner{f_V(A,A_0,A_1)}{(B,\Pi_0,\Pi_1)} \]
where the triples $(A,A_0,A_1)$ and $(B,\Pi_0,\Pi_1)$ have the form
\begin{alignat*}{3}
  A &: \cS_{01} \to \cA_{01} &\textrm{any stochastic} \\
  A_c &: \cS_c \to \cA_c &\textrm{any stochastic} &&&\qquad c\in\set{0,1}\\
  B &: \cT_{01} \to \cB_{01} &\textrm{no-signaling only} \\
  \Pi_c &: \cS_{01} \to \cA_c  \qquad\qquad & 0\leq\Pi_c\leq e_{\cA_c} \pi_\textrm{Alice}^* &&&\qquad c\in\set{0,1}.
\end{alignat*}
The linear mapping $f_V$ appearing in the inner product (and its adjoint) is defined by
\begin{align*}
  f_V &: (A,A_0,A_1) \mapsto
  \Pa{
    \Phi_V(A)\,,\,
    \mar{\cA_1}{A}-A_0\ot e_{\cS_1}^*\,,\,
    \mar{\cA_0}{A}-A_1\ot e_{\cS_0}^*
  } \\
  f_V^* &: (B,\Pi_0,\Pi_1) \mapsto
  \Pa{
    \Phi_V^*(B) + e_{\cA_1}\ot\Pi_0 + e_{\cA_0}\ot\Pi_1\,,\,
    -\Pi_0\Pa{I_{\cS_0}\ot e_{\cS_1}}\,,\,
    -\Pi_1\Pa{I_{\cS_1}\ot e_{\cS_0}}
  }
\end{align*}
so that
\[
  \Inner{f_V(A,A_0,A_1)}{(B,\Pi_0,\Pi_1)}
  = \Inner{V}{A\ot B} + \sum_{c\in\set{0,1}} \Inner{\mar{\cA_\ol{c}}{A}-A_c\ot e_{\cS_\ol{c}}^*}{\Pi_c}
\]
for all $(A,A_0,A_1)$ and all $(B,\Pi_0,\Pi_1)$.
(The adjoint mapping $f_V^*$ is not used until the algorithm of Figure \ref{fig:alg} and its proof of correctness in Proposition \ref{prop:lambda-alg}.)

\subsubsection*{Intuition}

Some explanation is in order.
As with the original min-max problem $\lambda(V)$, the matrices $A$ and $B$ represent the strategies employed by the teams.
Note, however, that in the definition of $\mu(V)$ Team Alice is now free to choose among arbitrary stochastic matrices for its strategy.
The matrices $A_0,A_1$ for Team Alice are purported witnesses to the claim that $A$ is a valid no-signaling matrix.

For the moment, we are concerned with relaxing the domain only of Team Alice's strategies, so Bob's strategy $B$ must still be no-signaling.
Bob's strategies will be addressed in Section \ref{sec:alg:oracle}.
The matrices $\Pi_0,\Pi_1$ for Team Bob are \emph{penalty matrices}---they are the means by which Team Bob penalizes Team Alice according to the extent that $A_0,A_1$ are false witnesses to the claim that $A$ is no-signaling.

The new objective function $\Inner{f_V(A,A_0,A_1)}{(B,\Pi_0,\Pi_1)}$ equals the old objective function $\inner{V}{A\ot B}$ plus two \emph{penalty terms}.
If $A$ is not a no-signaling matrix then the difference matrix
\[ \Delta_c \equiv \mar{\cA_\ol{c}}{A}-A_c\ot e_{\cS_\ol{c}}^* \]
must be nonzero for at least one $c$.
In this case, Bob selects $\Pi_c$ to pick out the positive entries of $\Delta_c$, which are then added the verifier's probability of rejection.

Let us informally explain why the restriction $0\leq\Pi_c\leq e_{\cA_c} \pi_\textrm{Alice}^*$ on penalty matrices is sufficient to remove Team Alice's incentive to cheat.
Suppose the $k_c$th entry of the $i$th column of the difference matrix $\Delta_c$ is a positive real number $\delta>0$ and suppose that $A'$ is a valid no-signaling matrix witnessed by $A_0,A_1$.
Since the verifier asks questions $i$ of Team Alice with probability $\pi_i$, it must be that, when selecting the probability with which to answer $k_c$, the advantage gained by Team Alice from using the inadmissible strategy $A$ instead of the no-signaling strategy $A'$ is at most $\delta\pi_i$.
By selecting a penalty matrix $\Pi_c$ so that the $k_c$th entry of the $i$th column of $\Pi_c$ is equal to $\pi_i$, Team Bob adds precisely the quantity $\delta\pi_i$ to the verifier's probability of rejection, thus eliminating the advantage obtained by Team Alice in acting according to $A$ instead of $A'$ for this particular choice of questions $i$ and answer $k_c$ from Alice$_c$.

Repeating this logic for all entries $(i,k_c)$ of $\Delta_c$, we find that Team Bob should select the penalty matrix $\Pi_c$ so that the $(i,k_c)$th entry is either zero or $\pi_i$ according to whether the corresponding entry of $\Delta_c$ is nonpositive or positive.
A penalty matrix of this form is called \emph{optimal for $(A,A_0,A_1)$} and satisfies
\[ \Inner{\Delta_c}{\Pi_c} = \Inner{\Delta_c^+}{e_{\cA_c} \pi_\textrm{Alice}^*} \]
where $\Delta_c^+$ is the positive part of $\Delta_c$.
(Here the \emph{positive part} of a real matrix $X$ is the matrix $X^+$ with the property that if $x$ is any entry of $X$ then the corresponding entry of $X^+$ is $\max\set{0,x}$.)

\subsection{Equivalence of the two min-max problems}
\label{sec:relaxation:equivalence}

We are now ready to prove the desired ``rounding theorem'' mentioned in the introduction, a corollary of which is the equivalence of the min-max problems $\mu(V)$ and $\lambda(V)$ (Corollary \ref{cor:rounding}).
The theorem employs two lemmas and their corollaries, the proofs of which appear below in Section \ref{sec:relaxation:lemmas}.

\begin{theorem}[Rounding theorem]
\label{thm:rounding}

Let $(A,A_0,A_1)$ be a feasible solution for $\mu(V)$ and let $\Pi_0^A,\Pi_1^A$ be optimal penalties for $(A,A_0,A_1)$.
There exists a no-signaling matrix $A_\ns$ witnessed by $A_0,A_1$ such that for all stochastic matrices $B$ it holds that
\[
  \Inner{V}{A_\ns\ot B}
  \leq \Inner{f_V(A,A_0,A_1)}{(B,\Pi_0^A,\Pi_1^A)}.
\]
Moreover, $A_\ns$ can be computed efficiently in parallel given $(A,A_0,A_1)$.

\end{theorem}

\begin{proof}

For both $c\in\set{0,1}$ let $\Delta_c^+$ be the positive part of
\( \mar{\cA_\ol{c}}{A} -  A_c\ot e_{\cS_\ol{c}}^* \)
and observe that
\[ \Delta_c^+ \leq \mar{\cA_\ol{c}}{A}. \]
By Corollary \ref{cor:nonnegative-difference} below there exists a preimage $D_0^+\geq 0$ of $\Delta_0^+$ with
\begin{align*}
  A-D_0^+&\geq 0 \\
  \mar{\cA_1}{D_0^+} &= \Delta_0^+.
\end{align*}
Let $\Gamma_1^+$ be the positive part of
\( \Mar{\cA_0}{A-D_0^+} - A_1\ot e_{\cS_0}^*. \)
As with $\Delta_c$ above, observe that
\[ \Gamma_1^+ \leq \Mar{\cA_0}{A-D_0^+}. \]
(Moreover, it is easy to see that $\Gamma_1^+ \leq \Delta_1^+$---a fact we employ later in this proof.)
Apply Corollary \ref{cor:nonnegative-difference} again to obtain a preimage $C_1^+\geq 0$ of $\Gamma_1^+$ with
\begin{align*}
  A-D_0^+-C_1^+&\geq 0 \\
  \mar{\cA_0}{C_1^+} &= \Gamma_1^+.
\end{align*}
Thus, we have a matrix $A-D_0^+-C_1^+\geq 0$ such that for both $c\in\set{0,1}$ it holds that
\[ \Mar{\cA_\ol{c}}{A-D_0^+-C_1^+} \leq A_c\ot e_{\cS_\ol{c}}^*. \]
Hence there exist nonnegative matrices $T_c:\cS_{01}\to\cA_c$ with
\[ \Mar{\cA_\ol{c}}{A-D_0^+-C_1^+} + T_c = A_c\ot e_{\cS_\ol{c}}^*. \]
Applying $\marginal_{A_c}$ to both sides of this equation we see that
\( \mar{\cA_0}{T_0} = \mar{\cA_1}{T_1}. \)
By Corollary \ref{cor:purify-no-sig} below there exists a nonnegative matrix $T:\cS_{01}\to\cA_{01}$
with \( \mar{\cA_\ol{c}}{T} = T_c \) for both $c\in\set{0,1}$.
The desired no-signaling matrix $A_\ns$ is given by
\[ A_\ns = A-D_0^+-C_1^++T. \]
As $D_0^+$, $C_1^+$, and $T$ can be computed efficiently in parallel, so too can $A_\ns$.
To see that $A_\ns$ is a no-signaling matrix witnessed by $A_0,A_1$ it suffices to observe that
\[ \mar{\cA_\ol{c}}{A_\ns} = \Mar{\cA_\ol{c}}{A-D_0^+-C_1^+} + T_c = A_c\ot e_{\cS_\ol{c}}^*. \]

It remains only to verify the stated inequality.
To this end, we have
\begin{align*}
  \Inner{V}{A_\ns\ot B}
  &= \Inner{A}{\Phi_V^*(B)} - \Inner{D_0^++C_1^+}{\Phi_V^*(B)} + \Inner{T}{\Phi_V^*(B)} \\
  &\leq \Inner{A}{\Phi_V^*(B)} + \Inner{T}{\Phi_V^*(B)} \\
  &\leq \Inner{A}{\Phi_V^*(B)} + \Inner{T}{e_{\cA_{01}} \pi_\textrm{Alice}^*}
\end{align*}
As $A_\ns$ and $A$ are both stochastic matrices, it must be that $D_0^++C_1^+$ and $T$ have the same column sums.
As $\inner{T}{e_{\cA_{01}} \pi_\textrm{Alice}^*}$ equals the sum of the column sums of $T$ weighted according to $\pi_\textrm{Alice}$, the matrix $T$ can be replaced by $D_0^++C_1^+$ without affecting this inner product.
That is
\[ \inner{T}{e_{\cA_{01}} \pi_\textrm{Alice}^*} = \inner{D_0^++C_1^+}{e_{\cA_{01}} \pi_\textrm{Alice}^*}. \]
Expanding the right side of this equality we obtain
\[
  \Inner{\mar{\cA_1}{D_0^+}}{e_{\cA_0}\pi_\textrm{Alice}^*} + \Inner{\mar{\cA_0}{C_1^+}}{e_{\cA_1}\pi_\textrm{Alice}^*}
  = \Inner{\Delta_0^+}{e_{\cA_0}\pi_\textrm{Alice}^*} + \Inner{\Gamma_1^+}{e_{\cA_1}\pi_\textrm{Alice}^*}.
\]
As $\Gamma_1^+\leq\Delta_1^+$ this quantity is at most
\[ \Inner{\Delta_0^+}{e_{\cA_0}\pi_\textrm{Alice}^*} + \Inner{\Delta_1^+}{e_{\cA_1}\pi_\textrm{Alice}^*}. \]
Putting everything together, we have
\begin{align*}
    \inner{V}{A_\ns\ot B} &\leq \inner{A}{\Phi_V^*(B)} +
    \Inner{\Delta_0^+}{e_{\cA_0}\pi_\textrm{Alice}^*} + \Inner{\Delta_1^+}{e_{\cA_1}\pi_\textrm{Alice}^*} \\
  &= \inner{A}{\Phi_V^*(B)} +
    \Inner{\mar{\cA_1}{A} - A_0\ot e_{\cS_1}^*}{\Pi_0^A} +
    \Inner{\mar{\cA_0}{A} - A_1\ot e_{\cS_0}^*}{\Pi_1^A} \\
  &= \Inner{f_V\pa{A,A_0,A_1}}{(B,\Pi_0^A,\Pi_1^A)}.
\end{align*}
as desired.
\end{proof}

\begin{corollary}[Equivalence of min-max problems]
\label{cor:rounding}

The following hold for any verifier $V$ and any $\delta\geq 0$:
\begin{enumerate}

\item \label{item:lambda:gv}
$\mu(V)=\lambda(V).$

\item \label{item:lambda:Bob}
If $(B^\mu,\Pi_0^\mu,\Pi_1^\mu)$ is $\delta$-optimal for $\mu(V)$ then $B^\mu$ is $\delta$-optimal for $\lambda(V)$.

\item \label{item:lambda:Alice}
If $(A^\mu,A_0^\mu,A_1^\mu)$ is $\delta$-optimal for $\mu(V)$ then there exists $A_\ns$ such that $A_\ns$ is $\delta$-optimal for $\lambda(V)$ and $A_\ns$ can be computed efficiently in parallel given $(A^\mu,A_0^\mu,A_1^\mu)$.

\end{enumerate}

\end{corollary}

\begin{proof}

We begin with item \ref{item:lambda:gv}.
It is easy to prove $\lambda(V)\geq\mu(V)$:
let $A^\lambda$ be optimal for $\lambda(V)$, let $A_0,A_1$ witness the fact that $A^\lambda$ is no-signaling, and let $(B^\mu,\Pi_0^\mu,\Pi_1^\mu)$ be optimal for $\mu(V)$.
Then
\[ \lambda(V) \geq \Inner{V}{A^\lambda\ot B^\mu} = \Inner{f_V(A^\lambda,A_0,A_1)}{(B^\mu,\Pi_0^\mu,\Pi_1^\mu)} \geq \mu(V). \]
For the reverse inequality, let $(A^\mu,A_0^\mu,A_1^\mu)$ be optimal for $\mu(V)$, let $\Pi_0^{A^\mu},\Pi_1^{A^\mu}$ be optimal penalties for $(A^\mu,A_0^\mu,A_1^\mu)$, and let $B^\lambda$ be optimal for $\lambda(V)$.
By Theorem \ref{thm:rounding} there exists a no-signaling matrix $A_\ns$ witnessed by $A_0^\mu,A_1^\mu$ such that
\[ \Inner{V}{A_\ns \ot B^\lambda} \leq \Inner{f_V\Pa{A^\mu,A_0^\mu,A_1^\mu}}{\Pa{B^\lambda,\Pi_0^{A^\mu},\Pi_1^{A^\mu}}}. \]
The desired inequality $\lambda(V)\leq\mu(V)$ follows from the fact that the left side is at least $\lambda(V)$ and the right side is at most $\mu(V)$.
The proof of item \ref{item:lambda:gv} is complete.

Item \ref{item:lambda:Bob} follows easily from item \ref{item:lambda:gv}.
Let $A$ be a no-signaling matrix and let $A_0,A_1$ witness this fact.
Then
\[ \lambda(V) - \delta = \mu(V) - \delta \leq \Inner{f_V(A,A_0,A_1)}{(B^\mu,\Pi_0^\mu,\Pi_1^\mu)} = \Inner{V}{A\ot B^\mu}. \]
As $A$ was chosen arbitrarily, it follows that $B^\mu$ is $\delta$-optimal for $\lambda(V)$.

For item \ref{item:lambda:Alice}, let $B$ be any no-signaling matrix and let $\Pi_0^{A^\mu},\Pi_1^{A^\mu}$ be optimal penalties for the given $\delta$-optimal solution $(A^\mu,A_0^\mu,A_1^\mu)$.
By Theorem \ref{thm:rounding} there exists a no-signaling matrix $A_\ns$ witnessed by $A_0^\mu,A_1^\mu$ such that
\[ \Inner{V}{A_\ns \ot B} \leq \Inner{f_V\Pa{A^\mu,A_0^\mu,A_1^\mu}}{\Pa{B,\Pi_0^{A^\mu},\Pi_1^{A^\mu}}} \leq \mu(V) + \delta = \lambda(V) + \delta. \]
As $B$ was chosen arbitrarily, it follows that $A_\ns$ is $\delta$-optimal for $\lambda(V)$.
\end{proof}

\subsection{Lemmas used in the rounding theorem}
\label{sec:relaxation:lemmas}

The lemmas used in the proof of Theorem \ref{thm:rounding} are not difficult.
It is quite likely that some form of these lemmas is part of computer science ``folklore,'' though our notation may be nonstandard.

\begin{lemma}[Small marginals have small preimages]
\label{lm:nonnegative-difference}

Let $a\in\cA_{01}$ and $\vec\delta\in\cA_0$ be nonnegative vectors with
\( \vec\delta \leq \mar{\cA_1}{a}. \)
There exists a nonnegative vector $d\in\cA_{01}$ with $d\leq a$ and $\mar{\cA_1}{d} = \vec\delta$.
Moreover, $d$ can be computed efficiently in parallel given $a,\vec\delta$.

\end{lemma}

\begin{proof}

Let $a_{(k_0,k_1)}$ and $\vec\delta_{k_0}$ denote the nonnegative entries of $a$ and $\vec\delta$, respectively.
Let $s_{k_0}$ denote the $k_0$th entry of $\mar{\cA_1}{a}$ so that
\[ s_{k_0} = \sum_{k_1=1}^{\dim(\cA_1)} a_{(k_0,k_1)}. \]
The desired vector $d$ has entries $d_{(k_0,k_1)}$ given by
\[ d_{(k_0,k_1)} =
  \left\{
    \begin{array}{ll}
      \displaystyle \vec\delta_{k_0}\frac{a_{(k_0,k_1)}}{s_{k_0}} & \textrm{when $s_{k_0}\neq 0$} \\
      0 & \textrm{otherwise}
    \end{array}
  \right.
\]
(Intuitively, the weight $\vec\delta_{k_0}$ required of $\sum_{k_1} d_{(k_0,k_1)}$ is ``spread out'' over each $d_{(k_0,k_1)}$ proportionately according to $a_{(k_0,k_1)}$.)
It is clear that this construction can be implemented efficiently in parallel.

Let us verify that $d\leq a$.
Observe that for the case $s_{k_0}\neq 0$ the ratio $\vec\delta_{k_0}/s_{k_0}$ is at most one because $\vec\delta\leq \mar{\cA_1}{a}$.
Then \[ d_{(k_0,k_1)} = a_{(k_0,k_1)} \frac{\vec\delta_{k_0}}{s_{k_0}} \leq a_{(k_0,k_1)} \]
as desired.
Of course, if $s_{k_0}= 0$ then $d_{(k_0,k_1)}=0$ by definition and hence $d_{(k_0,k_1)}\leq a_{(k_0,k_1)}$ because $a\geq 0$.

Let us verify that $\mar{\cA_1}{d} = \vec\delta$.
For the case $s_{k_0}\neq 0$ the $k_0$th entry of $\mar{\cA_1}{d}$ is given by
\[ \sum_{k_1=1}^{\dim(A_1)} d_{(k_0,k_1)} = \frac{\vec\delta_{k_0}}{s_{k_0}} \sum_{k_1=1}^{\dim(A_1)} a_{(k_0,k_1)} = \vec\delta_{k_0} \]
as desired.
As above, if $s_{k_0}=0$ then by definition $d_{(k_0,k_1)}=0$ for each $k_1$ and hence $\sum_{k_1} d_{(k_0,k_1)} = 0$.
As
$0\leq\vec\delta_{k_0}\leq s_{k_0}$
it must be that $\vec\delta_{k_0}=0$, too.
\end{proof}

\begin{corollary}
\label{cor:nonnegative-difference}

Let $A:\cS_{01}\to\cA_{01}$ and $\Delta:\cS_{01}\to\cA_0$ be nonnegative matrices with
\( \Delta \leq \mar{\cA_1}{A}. \)
There exists a nonnegative matrix $D:\cS_{01}\to\cA_{01}$ with $D\leq A$ and $\mar{\cA_1}{D} = \Delta$.
Moreover, $D$ can be computed efficiently in parallel given $A,\Delta$.

\end{corollary}

\begin{proof}

Apply Lemma \ref{lm:nonnegative-difference} to each of the columns of $A,\Delta$.
\end{proof}

\begin{lemma}[Disjoint marginals are always consistent]
\label{lm:purify-no-sig}

For both $c\in\set{0,1}$ let $t_c\in\cA_c$ be nonnegative vectors whose entries sum to the same value.
There exists a nonnegative vector $t\in\cA_{01}$ with $\mar{\cA_\ol{c}}{t}=t_c$ for both $c\in\set{0,1}$.
Moreover, $t$ can be computed efficiently in parallel given $t_0,t_1$.

\end{lemma}

\begin{proof}

Let $p_{k_0}$ and $q_{k_1}$ be the nonnegative entries of $t_0$ and $t_1$, respectively.
Let $s$ denote the sum of the entries of $t_0,t_1$ so that
\[ s = \sum_{k_0=1}^{\dim(\cA_0)} p_{k_0} = \sum_{k_1=1}^{\dim(\cA_1)} q_{k_1}. \]
If $s=0$ then it is clear that the desired vector $t$ is the zero vector.
For the remainder of the proof assume that $s\neq 0$.
The desired vector $t$ has entries $t_{(k_0,k_1)}$ given by
\[ t_{(k_0,k_1)} = \frac{p_{k_0} q_{k_1}}{s} \]
It is clear that this construction can be implemented efficiently in parallel.

Let us verify that $\mar{\cA_\ol{c}}{t} = t_c$ for both $c\in\set{0,1}$.
For the case $c=0$ the $k_0$th entry of $\mar{\cA_1}{t}$ is given by
\[ \sum_{k_1=1}^{\dim(\cA_1)} \frac{ p_{k_0} q_{k_1}}{s} = \frac{ p_{k_0} s}{s} = p_{k_0} \]
as desired.
The case $c=1$ is handled similarly.
\end{proof}

\begin{corollary}
\label{cor:purify-no-sig}

For both $c\in\set{0,1}$ let $T_c:\cS_{01}\to\cA_c$ be nonnegative matrices with $\mar{\cA_0}{T_0} = \mar{\cA_1}{T_1}$.
There exists a nonnegative matrix $T:\cS_{01}\to\cA_{01}$ with $\mar{\cA_\ol{c}}{T} = T_c$ for both $c\in\set{0,1}$.
Moreover, $T$ can be computed efficiently in parallel given $T_0,T_1$.

\end{corollary}

\begin{proof}

Apply Lemma \ref{lm:purify-no-sig} to each of the columns of $T_0,T_1$.
\end{proof}

\section{A parallel multiplicative weights algorithm}
\label{sec:alg}

In this section we complete the proof of our main result---that every decision problem that admits a two-turn interactive proof with competing teams of no-signaling provers is also in $\cls{PSPACE}$.
Most of the detail appears in Section \ref{sec:alg:alg} wherein we present an efficient parallel oracle-algorithm based on the MWUM that produces $\delta$-optimal no-signaling strategies for the teams, given
an oracle for ``best responses'' for Team Bob to a given candidate strategy for Alice.
We describe an efficient parallel implementation of the required oracle in Section \ref{sec:alg:oracle}, from which the unconditional efficiency of our algorithm immediately follows.
The ensuing inclusion of $\cls{MRG}_\ns(2,2)$ inside $\cls{PSPACE}$ is discussed in Section \ref{sec:alg:PSPACE}.

\subsection{The parallel algorithm}
\label{sec:alg:alg}

Precise statements of the problem solved by our algorithm and the oracle it requires are given below.
All input numbers are written as rational numbers in binary.
For matrix inputs, each entry is written explicitly.

\begin{problem}
  [Weak no-signaling equilibrium]
  \label{problem:lambda}
  \ \\[1mm]
  \begin{tabularx}{\textwidth}{lX}
    \emph{Input:} &
    A verifier matrix $V:\cS_{01}\cT_{01}\to\cA_{01}\cB_{01}$ and an accuracy parameter $\delta>0$.
    \\[1mm]
    \emph{Oracle:} &
    Weak no-signaling optimization.
    (See Problem \ref{problem:oracle} below.)
    \\[1mm]
    \emph{Output:} &
    $\delta$-optimal no-signaling strategies $\tilde{A},\tilde{B}$ for the min-max problem $\lambda(V)$.
  \end{tabularx}
\end{problem}

\begin{problem}
  [Weak no-signaling optimization]
  \label{problem:oracle}
  \ \\[1mm]
  \begin{tabularx}{\textwidth}{lX}
  \emph{Input:} &
  A verifier-Alice matrix $S:\cT_{01}\to\cB_{01}$ and an accuracy parameter $\delta>0$.
  \\[1mm]
  \emph{Output:} &
  A $\delta$-optimal no-signaling strategy $\tilde{B}$ for Team Bob.
  (That is, a no-signaling matrix $\tilde{B}$ such that $\inner{S}{\tilde{B}}\geq\inner{S}{B}-\delta$ for all no-signaling matrices $B$.)
  \end{tabularx}
\end{problem}

Given Corollary \ref{cor:rounding}, it suffices to find $\delta$-optimal solutions
$(\tilde{A},\tilde{A}_0,\tilde{A}_1)$ and $(\tilde{B},\tilde{\Pi}_0,\tilde{\Pi}_1)$
for $\mu(V)$ and then convert these solutions into $\delta$-optimal strategies for $\lambda(V)$.
This method is codified in the algorithm of Figure \ref{fig:alg}.

This algorithm is a straightforward modification of the standard multiplicative weights update method for equilibrium problems.
The precise formulation of the MWUM used in this paper is stated as Theorem \ref{thm:mwum}.
Our statement of this theorem is somewhat nonstandard: the result is usually presented in the form of an algorithm, whereas our presentation is purely mathematical.
However, a cursory examination of the literature---say, Kale's thesis \cite[Chapter 2]{Kale07}---reveals that our mathematical formulation is equivalent to the more conventional algorithmic form.

\begin{theorem}[Multiplicative weights update method---see Ref.~{\cite[Theorem 2]{Kale07}}]
\label{thm:mwum}

Fix an $\varepsilon\in(0,1/2)$.
Let $m^1,\dots,m^T$ be arbitrary $D$-dimensional ``loss'' vectors whose entries $m^t_i$ lay in the interval $[-\alpha,\alpha]$.
Let $w^1,\dots,w^T$ be $D$-dimensional nonnegative ``weight'' vectors whose entries $w^t_i$ are given recursively via
\begin{align*}
  w^1_i &= 1 \\
  w^{t+1}_i &= w^t_i \Pa{1-\varepsilon m^t_i}.
\end{align*}
Let $p^1,\dots,p^T$ be probability vectors obtained by normalizing each $w^1,\dots,w^T$.
For all probability vectors $p$ it holds that
\[
  \frac{1}{T}\sum_{t=1}^T \Inner{p^t}{m^t} \leq
  \Inner{p}{\frac{1}{T}\sum_{t=1}^T m^t} + \alpha\Pa{\varepsilon + \frac{\ln D}{\varepsilon T}}.
\]

\end{theorem}

Note that Theorem \ref{thm:mwum} holds for \emph{all} choices of loss vectors $m^1,\dots,m^T$, including the case in which each $m^t$ is chosen adversarially based upon $w^t$.
This adaptive selection of loss vectors is typical in implementations of the MWUM.

\begin{figure}[p]
\hrule
\begin{enumerate}

\item Let $\varepsilon=\delta/10$ and let $T=\left\lceil \frac{\ln(\dim(\cA_{01}))}{\varepsilon^2} \right\rceil$.

Let
\(
  \Pa{W^1,W_0^1,W_1^1}
\)
denote the triple of all-ones matrices
and let
\( \Pa{A^1,A_0^1,A_1^1} \) denote the uniformly random strategy for Alice obtained by normalizing the columns of $\Pa{W^1,W_0^1,W_1^1}$.

\item
Repeat for each $t=1,\dots,T$:
\begin{enumerate}

\item
Compute optimal penalties $\Pi_0^t,\Pi_1^t$ for $(A^t,A_0^t,A_1^t)$ as described in Section \ref{sec:relaxation:mu}.
Use the oracle for Problem \ref{problem:oracle} to obtain a $\delta/2$-best response $B^t$ to the verifier-Alice matrix $\Phi_V(A^t)$.

\item
Compute the loss matrices $\Pa{M^t,M_0^t,M_1^t} = f_V^*\Pa{B^t,\Pi_0^t,\Pi_1^t}$.
Exit the loop now if $t=T$.

\item
Update the weight matrices according to the standard multiplicative weights update rule:
\[ \Pa{W^{t+1},W_0^{t+1},W_1^{t+1}} = \Pa{W^t,W_0^t,W_1^t}\boxtimes\Pa{\underbrace{\Pa{W^1,W_0^1,W_1^1}}_{\textrm{all-ones matrices}}-\varepsilon \Pa{M^t,M_0^t,M_1^t}} \]
where $\boxtimes$ denotes the (entrywise) matrix Schur product.
(See Theorem \ref{thm:mwum}.)

\item
Compute the updated triple $(A^{t+1},A_0^{t+1},A_1^{t+1})$ of stochastic matrices for Team Alice by normalizing the columns of $(W^{t+1},W_0^{t+1},W_1^{t+1})$.

\end{enumerate}

\item
Compute
\[
  (\tilde{A},\tilde{A}_0,\tilde{A}_1)=\frac{1}{T}\sum_{t=1}^T (A^t,A_0^t,A_1^t)
  \qquad \textrm{and} \qquad
  (\tilde{B},\tilde{\Pi}_0,\tilde{\Pi}_1) = \frac{1}{T} \sum_{t=1}^T (B^t,\Pi_0^t,\Pi_1^t)
\]
both of which are $\delta$-optimal for $\mu(V)$.
Compute the no-signaling matrix $\tilde{A}_\ns$ from $(\tilde{A},\tilde{A}_0,\tilde{A}_1)$ as described in Corollary \ref{cor:rounding}.

\item
Return $(\tilde{A}_\ns,\tilde{B})$ as the $\delta$-optimal strategies of Team Alice and Team Bob for $\lambda(V)$.

\end{enumerate}
\hrule
\caption{Algorithm that finds $\delta$-optimal solutions to the equilibrium problem $\lambda(V)$ for two-turn interactive proofs with competing teams of no-signaling provers (Problem \ref{problem:lambda}).}
\label{fig:alg}
\end{figure}

\begin{proposition}
\label{prop:lambda-alg}

The oracle-algorithm presented in Figure \ref{fig:alg} solves the weak no-signaling equilibrium problem (Problem \ref{problem:lambda}).
Assuming unit cost for the oracle, this algorithm can be implemented in parallel with run time bounded by a polynomial in $1/\delta$ and $\log(\dim(\cS_{01}\cT_{01}\cA_{01}\cB_{01}))$.

\end{proposition}

\begin{proof}

Let $m^t$ denote the $i$th column of $M^t$ for each $t=1,\dots,T$.
We argue that the entries of $m^t$ lay in the interval $[0,3\pi_i]$.
To this end, observe that the loss matrix $M^t$ is defined in Figure \ref{fig:alg} via the adjoint mapping $f_V^*$ as
\[ M^t = \Phi_V^*(B^t) + e_{\cA_1}\ot\Pi_0^t + e_{\cA_0}\ot\Pi_1^t \leq 3e_{\cA_{01}}\pi_\textrm{Alice}^* \]
where the inequality follows immediately from the bound $\Phi_V^*(B)\leq e_{\cA_{01}}\pi_\textrm{Alice}^*$ of Proposition \ref{prop:Phi-bound} and the restriction $\Pi_c\leq e_{\cA_c}\pi_\textrm{Alice}^*$ on penalty matrices.
The desired bound on the entries of $m^t$ follows from the observation that the $i$th column of $3e_{\cA_{01}}\pi_\textrm{Alice}^*$ is the vector whose entries are all equal to $3\pi_i$.

Let $a^t$ denote the $i$th column of $A^t$ for $t=1,\dots,T$.
It is clear that the construction of the probability vectors $a^t$ in terms of the loss vectors $m^t$ presented in Figure \ref{fig:alg} obeys the condition of Theorem \ref{thm:mwum}.
It therefore follows that for any probability vector $a\in\cA_{01}$ we have
\[
  \frac{1}{T}\sum_{t=1}^T \Inner{a^t}{m^t} \leq
  \Inner{a}{\frac{1}{T}\sum_{t=1}^T m^t} + 3\pi_i\Pa{\varepsilon + \frac{\ln(\dim(\cA_{01}))}{\varepsilon T}}.
\]
Summing these inequalities over all columns $i$ we find that for any stochastic matrix $A$ it holds that
\[
  \frac{1}{T}\sum_{t=1}^T \Inner{A^t}{M^t} \leq
  \Inner{A}{\frac{1}{T}\sum_{t=1}^T M^t} + 3\Pa{\varepsilon + \frac{\ln(\dim(\cA_{01}))}{\varepsilon T}}.
\]

A similar bound on the stochastic matrices $A_0^t,A_1^t$ in terms of the loss matrices $M_0^t,M_1^t$ can be derived in much the same way.
For completeness, let us make this argument explicit.
For both $c\in\set{0,1}$ and for each question $i_c$ let $\pi_{i_c}$ denote the probability with which the referee asks question $i_c$ to Alice$_c$.
Let $m_c^t$ denote the $i_c$th column of $M_c^t$ for each $t=1,\dots,T$.
We argue that the entries of $m_c^t$ lay in the interval $[-\pi_{i_c},0]$.
Recall the loss matrix $M_c^t$ is defined in Figure \ref{fig:alg} via the adjoint mapping $f_V^*$ as
\[ M_c^t = -\Pi_c^t \Pa{I_{\cS_c}\ot e_{\cS_\ol{c}}} \geq -e_{\cA_c}\mar{\cS_\ol{c}}{\pi_\textrm{Alice}}^* \]
where the inequality follows immediately from the restriction $\Pi_c\leq e_{\cA_c}\pi_\textrm{Alice}^*$ on penalty matrices.
The desired bound on the entries of $m_c^t$ follows from the observation that the $i_c$th column of $e_{\cA_c}\mar{\cS_\ol{c}}{\pi_\textrm{Alice}}^*$ is the vector whose entries are all equal to $\pi_{i_c}$.

As above, let $a_c^t$ denote the $i_c$th column of $A_c^t$ for $t=1,\dots,T$.
It is clear that the construction of the probability vectors $a_c^t$ in terms of the loss vectors $m_c^t$ presented in Figure \ref{fig:alg} obeys the condition of Theorem \ref{thm:mwum}.
It therefore follows that for any probability vector $a_c\in\cA_c$ we have
\[
  \frac{1}{T}\sum_{t=1}^T \Inner{a_c^t}{m^t} \leq
  \Inner{a_c}{\frac{1}{T}\sum_{t=1}^T m_c^t} + \pi_{i_c}\Pa{\varepsilon + \frac{\ln(\dim(\cA_c))}{\varepsilon T}}.
\]
Summing these inequalities over all columns $i_c$ we find that for any stochastic matrix $A_c$ it holds that
\[
  \frac{1}{T}\sum_{t=1}^T \Inner{A_c^t}{M_c^t} \leq
  \Inner{A_c}{\frac{1}{T}\sum_{t=1}^T M_c^t} + \varepsilon + \frac{\ln(\dim(\cA_c))}{\varepsilon T}.
\]

At this point we have derived three inequalities for three arbitrary stochastic matrices $A,A_0,A_1$.
Summing these inequalities and substituting $(M^t,M_0^t,M_1^t)=f_V^*(B^t,\Pi_0^t,\Pi_1^t)$ and the choices of $\varepsilon,T$ listed in Figure \ref{fig:alg} we find that for any triple $(A,A_0,A_1)$ of stochastic matrices it holds that
\begin{equation}
\label{eq:mu-bound}
  \frac{1}{T}\sum_{t=1}^T \Inner{f_V(A^t,A_0^t,A_1^t)}{(B^t,\Pi_0^t,\Pi_1^t)} \leq
  \Inner{f_V(A,A_0,A_1)}{\frac{1}{T}\sum_{t=1}^T (B^t,\Pi_0^t,\Pi_1^t)} + \delta/2.
\end{equation}

The remainder of this proof is a straightforward adaptation of Kale's analysis for the much simpler class of two-player zero-sum games in normal form \cite[Section 2.3.1]{Kale07}.
We argue that the triples  $(\tilde{A},\tilde{A}_0,\tilde{A}_1)$ and $(\tilde{B},\tilde{\Pi}_0,\tilde{\Pi}_1)$ appearing in Figure \ref{fig:alg} are $\delta$-optimal for $\mu(V)$.
Let us begin with the triple $(\tilde{A},\tilde{A}_0,\tilde{A}_1)$.
Choose any $(B,\Pi_0,\Pi_1)$ and let $(A^\star,A_0^\star,A_1^\star)$ be optimal for $\mu(V)$.
We have
\begin{align*}
  \Inner{\frac{1}{T}\sum_{t=1}^T f_V(A^t,A_0^t,A_1^t)}{(B,\Pi_0,\Pi_1)}
  &\leq \frac{1}{T}\sum_{t=1}^T \Inner{f_V(A^t,A_0^t,A_1^t)}{(B^t,\Pi_0^t,\Pi_1^t)} + \delta/2 \\
  &\leq \Inner{f_V(A^\star,A_0^\star,A_1^\star)}{\frac{1}{T}\sum_{t=1}^T (B^t,\Pi_0^t,\Pi_1^t)} + \delta \leq \mu(V)+\delta
\end{align*}
as desired.
(The first inequality is because each $(B^t,\Pi_0^t,\Pi_1^t)$ is a $\delta/2$-best response to $(A^t,A_0^t,A_1^t)$; the second is Eq.~\eqref{eq:mu-bound}.)

To see that $(\tilde{B},\tilde{\Pi}_0,\tilde{\Pi}_1)$ is $\delta$-optimal for $\mu(V)$, let $(A,A_0,A_1)$ be any triple of stochastic matrices.
We have
\[
  \Inner{f_V(A,A_0,A_1)}{\frac{1}{T}\sum_{t=1}^T(B^t,\Pi_0^t,\Pi_1^t)}
  \geq \frac{1}{T}\sum_{t=1}^T \Inner{f_V(A^t,A_0^t,A_1^t)}{(B^t,\Pi_0^t,\Pi_1^t)} -\delta/2
  \geq \mu(V) - \delta
\]
as desired.
(The first inequality is Eq.~\eqref{eq:mu-bound}; the second is because each $(B^t,\Pi_0^t,\Pi_1^t)$ is a $\delta/2$-best response to $(A^t,A_0^t,A_1^t)$.)
Finally, it follows from Corollary \ref{cor:rounding} that $\tilde{A}_\ns$ and $\tilde{B}$ are $\delta$-optimal strategies for $\lambda(V)$.

That the algorithm admits an efficient parallel implementation is straightforward.
In each iteration computations of optimal penalties, the loss matrices (via $f_V^*$), the multiplicative weights update rule, and normalization are all simple operations involving only addition and multiplication of individual rational entries of matrices that can easily be implemented in parallel.
Efficiency follows from the fact that the total number of iterations is bounded by a polynomial in $1/\delta$ and the logarithm of $\dim(\cS_{01}\cT_{01}\cA_{01}\cB_{01})$, the size of the verifier matrix.
\end{proof}

\subsection{Implementations of the best-response oracle for Team Bob}
\label{sec:alg:oracle}

In order for the algorithm of Figure \ref{fig:alg} to be unconditionally efficient, we require a parallel implementation of the oracle for weak no-signaling optimization (Problem \ref{problem:oracle}).
Fortunately, all the work is already done: Problem \ref{problem:oracle} is the optimization problem that arises naturally from two-turn, two-prover interactive proofs with no-signaling provers.
Thus, the parallel algorithm of Ito  \cite{Ito09} can be re-used to implement the oracle in our algorithm without complication.

In Ito's terminology, the verifier-Alice matrix $\Phi_V(A)$ specifies a \emph{game} and the two no-signaling provers comprising Team Bob are the \emph{players}.
Ito does not claim that an explicit strategy for the players can be found efficiently in parallel.
Rather, he claims only that the task of distinguishing high success probability from low success probability admits a parallel algorithm, as this simpler task is sufficient to put $\cls{MIP}_\ns(2,2)$ inside $\cls{PSPACE}$.
However, a cursory glance at the details of Ito's proof reveals a parallel construction of near-optimal no-signaling strategies for the players as required by Problem \ref{problem:oracle}.

Alternatively, the oracle for weak no-signaling optimization (Problem \ref{problem:oracle}) can be implemented by re-using the algorithm for weak no-signaling equilibrium (Problem \ref{problem:lambda}) listed in Figure \ref{fig:alg} of the present paper.
Indeed, Problem \ref{problem:oracle} is a special case of Problem \ref{problem:lambda} in which one team has a trivial strategy space.
In this special case the required ``oracle'' demands only weak no-signaling optimization over a trivial strategy space, which of course admits a trivial parallel implementation.
In other words, the algorithm of Figure \ref{fig:alg} can be used in a two-level recursive fashion to give an unconditionally efficient parallel algorithm for Problem \ref{problem:lambda}.

\subsection{Containment in PSPACE}
\label{sec:alg:PSPACE}

The desired containment of $\cls{MRG}_\ns(2,2)$ inside $\cls{PSPACE}$ now follows in the usual way:

\begin{numberedtheorem}{\ref{thm:main-result}}

Every decision problem that admits a two-turn interactive proof with competing teams of two no-signaling provers per team is also in $\cls{PSPACE}$.
Thus, we obtain the identity $\cls{MRG}_\ns(2,2)=\cls{PSPACE}$.

\end{numberedtheorem}

\begin{proof}

Let $L$ be a decision problem in $\cls{MRG}_\ns(2,2)$ with completeness $c$ and soundness $s$ and let $x$ be any input string.
Each entry of the exponential-size verifier matrix $V:\cS_{01}\cT_{01}\to\cA_{01}\cB_{01}$ induced by the verifier on input $x$ can be computed in space polynomial in $|x|$ by simulating every choice of randomness for the verifier.
In order to decide whether $x$ is a yes-instance or no-instance of $L$ it suffices to find $\delta$-optimal strategies for the teams for $\delta = \pa{c-s}/3$, which permits us to distinguish $\lambda(V)\geq c$ from $\lambda(V)\leq s$.
It follows from Proposition \ref{prop:lambda-alg} and the discussion in Section \ref{sec:alg:oracle} that the algorithm of Figure \ref{fig:alg} can be used to find $\delta$-optimal strategies for the teams and can be implemented in parallel with run time bounded by a polynomial in $1/\delta$ and the logarithm of the dimensions of $V$.
As the dimensions of $V$ scale exponentially with $|x|$ and $\delta$ scales as an inverse polynomial in $|x|$ the total run time of this parallel algorithm scales polynomially with $|x|$ and can therefore be simulated in polynomial space in the usual way \cite{Borodin77}.
\end{proof}

\section{Open problems, limitations of the present approach}
\label{sec:conclusion}

Attention is restricted in this paper to interactions with no more than two no-signaling provers per team and no more than two messages exchanged with each prover.
The purpose for this restriction, quite simply, is that this class of interactions appears to be the largest to which our techniques apply.

For all we know, interactions with three messages for a prover or three provers on a team could be sufficiently powerful to capture all of $\cls{EXP}$.
Indeed, it is consistent with current knowledge that a three-message protocol for $\cls{EXP}$ might require only \emph{one} prover per team, or that a three-prover no-signaling protocol for $\cls{EXP}$ might require only \emph{one} team of provers.
Given this paucity of upper bounds for similar, seemingly weaker models it is hoped that any reservation at the restrictions in our model is more than compensated by the fact that we are able to say anything at all about it.

Let us list some natural extensions of the two-prover, two-turn model and point out exactly where our method fails for these extensions.
\begin{description}

\item[More than two turns, only one prover per team.]
Perhaps the most important open problem related to our work is the complexity of $k$-turn interactive proofs with competing provers for constants $k\geq 3$.
This problem, which dates back at least to 1997 \cite{FeigeK97}, is still open even in the special case of only one prover per team.
With only one prover per team, the question is really a game-theoretic question with a much wider application than just interactive proofs.

Our method fails for this case because we do not have a bound on the verifier matrix of the form $V \leq e_{\cA_{01}\cB_{01}} \pi^*$ such as that appearing in Proposition \ref{prop:Phi-bound}.
Thus, we do not obtain a good enough bound on the loss vectors appearing in our variant of the multiplicative weights update method.

\item[More than two turns, only one team of no-signaling provers.]
The complexity of $k$-turn multi-prover interactive proofs with two no-signaling provers is still open for $k\geq 3$, even with only one team of provers \cite{Ito09}.
For ordinary multi-prover interactive proofs---in which the provers are not allowed to implement arbitrary no-signaling strategies---it is known that a multi-turn protocol with any number of provers can be simulated by another protocol with only two turns and two provers \cite{FeigeL92}.

Our method fails here for the same reason as above---that we cannot bound the loss vectors in the multiplicative weights update method for a multi-turn verifier.

\item[More than two provers, only one team of no-signaling provers.]
Similarly, the complexity of two-turn multi-prover interactive proofs with more than two no-signaling provers is still open, even with only one team of provers \cite{Ito09}.
As mentioned above, ordinary multi-prover interactive proofs require only two provers \cite{FeigeL92}.

Our method does not extend to this case either, as there is no known analogue of Lemma \ref{lm:purify-no-sig} for more than two provers.

\item[Quantum verifier and/or provers.]
Even with two no-signaling provers, two turns of interaction, and only one team of provers, it is still not known that the $\cls{PSPACE}$ upper bound holds when either the verifier or provers can send quantum messages \cite{Ito09}.
Here the problem is that Lemma \ref{lm:purify-no-sig} does not hold for quantum states.

\end{description}

\section*{Acknowledgements}

The author is grateful to Tsuyoshi Ito, Sarvagya Upadhyay, John Watrous, and Xiaodi Wu for helpful discussions.
This research was conducted while the author was a postdoc at the Institute for Quantum Computing and School of Computer Science at the University of Waterlo in Waterloo, Ontario, Canada, at which time the author was supported by the Government of Canada through Industry Canada, the Province of Ontario through the Ministry of Research and Innovation, NSERC, DTO-ARO, CIFAR, and QuantumWorks.

\newcommand{\etalchar}[1]{$^{#1}$}

\end{document}